\documentclass[twocolumn]{autart}    %

\usepackage{graphicx}          %

\usepackage{amsmath}
\usepackage{amsfonts}       %
\usepackage{amssymb}
\usepackage{color}

\newcommand{\set}[1]{\left\{#1\right\}}
\newcommand{\norm}[1]{\left\Vert #1 \right\Vert}
\newcommand{\abs}[1]{\left\vert #1 \right\vert}
\newcommand{\ra}{\rightarrow}
\newcommand{\Real}{\mathbb{R}}
\newcommand{\eps}{\varepsilon}

\newcommand{\B}{\mathcal{B}}
\renewcommand{\S}{\mathcal{S}}
\renewcommand{\subset}{\subseteq}
\newcommand{\R}{\mathcal{R}}
\newcommand{\Rd}{\mathcal{R}_\delta}
\newcommand{\Sd}{\mathcal{S}_\delta}
\newcommand{\Wd}{\mathcal{W}_\delta}
\newcommand{\Gd}{\mathcal{G}_\delta}
\newcommand{\Pd}{\Phi_\delta}
\newcommand{\dx}{\dot{x}}
\usepackage{xcolor}
\newcommand{\ymmark}[1]{{\color{black} #1}}
\newcommand{\U}{\mathcal{U}}
\newcommand{\ub}{\mathbf{u}}

\newcommand{\BlackBox}{\rule{1.5ex}{1.5ex}}  %
\newenvironment{proof}{\par\noindent{\bf Proof\ }}{\hfill\BlackBox\\[2mm]}

\begin{document}

\begin{frontmatter}

\title{Smooth Converse Lyapunov-Barrier Theorems for Asymptotic Stability with Safety Constraints and Reach-Avoid-Stay Specifications\thanksref{footnoteinfo}} %

\thanks[footnoteinfo]{This paper was not presented at any IFAC 
meeting. Corresponding author: Jun Liu. Tel. +1 519-888-4567 ext. 37550. Fax +1 519-746-4319.}

\author{Yiming Meng} \ead{yiming.meng@uwaterloo.ca}, 
\author{Yinan Li}\ead{yinan.li@uwaterloo.ca},  
\author{Maxwell Fitzsimmons}\ead{mfitzsimmons@uwaterloo.ca}, and 
\author{Jun Liu}\ead{j.liu@uwaterloo.ca}

\address{Department of Applied Mathematics\\
       University of Waterloo\\
       Waterloo, Ontario N2L 3G1, Canada\\}  %


\begin{abstract}                          %
Stability and safety are two important aspects in safety-critical control of dynamical systems. It has been a well established fact in control theory that stability properties can be characterized by Lyapunov functions. Reachability properties can also be naturally captured by Lyapunov functions for finite-time stability. Motivated by safety-critical control applications, such as in autonomous systems and robotics, there has been a recent surge of interests in characterizing safety properties using barrier functions. Lyapunov and barrier functions conditions, however, are sometimes viewed as competing objectives. In this paper, we provide a unified theoretical treatment of Lyapunov and barrier functions in terms of converse theorems for stability properties with safety guarantees and reach-avoid-stay type specifications. We show that if a system (modeled as a  dynamical system \ymmark{with measurable perturbations}) possesses a stability with safety property, then there exists a smooth Lyapunov function to certify such a property. This Lyapunov function is shown to be defined on the entire set of initial conditions from which solutions satisfy this property. A similar but slightly weaker statement is made for reach-avoid-stay specifications. We show by a simple example that the latter statement cannot be strengthened without additional assumptions. {We further extend the results for systems with control inputs and prove existence of converse Lyapunov-barrier functions for reach-and-avoid specifications.} \ymmark{One clear limitation of the results of this paper is that the converse Lyapunov-barrier theorems are not constructive, as with classical converse Lyapunov theorems. We believe, however, that the unified necessary and sufficient conditions with a single Lyapunov-barrier function are of theoretical interest and can hopefully shed some light on computational approaches.}

\end{abstract}

\end{frontmatter}

\section{Introduction}

Lyapunov stability theory \cite{lyapunov1992general} has been a cornerstone of automatic control. It is well known that stability properties for various models of nonlinear systems can be characterized by Lyapunov functions in the form of converse Lyapunov theorems \cite{massera1949liapounoff,wilson1969smoothing,kurzweil1956inversion,lin1996smooth,clarke1998asymptotic,teel2000smooth} (see also \cite[Section 1.1]{teel2000smooth} for a nice historical account). The corresponding Lyapunov conditions characterize the regularity of dynamical flows and provide crucial criteria for stability analysis. \ymmark{Assuming local Lipschitz continuity of the system dynamics, one can show that such Lyapunov functions possess additional nice properties such as smoothness  %
that can be used to infer robust stability properties \cite{lin1996smooth}. The same idea also constitutes the underlining philosophy of applying Lyapunov methods to control design \cite{freeman2008robust}}. 

In recent years, safety properties for dynamical systems received considerable attention, primarily motivated by safety-critical control applications, such as in autonomous cyber-physical systems and robotics \cite{agrawal2017discrete,cheng2019end,hsu2015control,nguyen2020dynamic,ames2019control,yang2019self}. In these applications, barrier functions \cite{prajna2004safety} are  used to certify that solutions of a given system can stay within a prescribed safe set, along with their control variants, called control barrier functions \cite{wieland2007constructive,ames2016control}, to provide feedback controls that render the system safe. The barrier function approach can be further combined with Lyapunov method to satisfy stability and safety requirements at the same  \cite{tee2009barrier,ames2016control,wu2019control,romdlony2016stabilization,jankovic2017combining,niu2013barrier}. Such formulations are amenable to optimization-based solutions enabled by quadratic programming \cite{ames2016control,xu2015robustness} or model predictive control \cite{wu2019control}, provided that the control system is in a control-affine form. 

Another important characteristic of a dynamical system is whether or not its solutions can reach a certain target set from a given initial set ({with or without control}). This is defined as reachability, which plays a key role in analysis and {in particular} control of dynamical systems \cite{bertsekas1971minimax}. Reachability analysis and control can also be viewed as an important special case of verification and control of dynamical systems with respect to more general formal specifications \cite{li2020robustly}. Since asymptotic stability entails asymptotic attraction, reachability can be naturally captured by asymptotic stability and Lyapunov conditions. 

The stability/reachability and safety objectives, however, are sometimes conflicting. For example, while a system can reach a target set from a given initial set, it may have to traverse an unsafe region to do so. For this reason, when formulating the problem as an optimization problem, some authors defined safety as a hard constraint, and reachability/stability as a soft (performance) requirements \cite{ames2016control}. 

The main objective of this paper is to provide a theoretical perspective on uniting Lyapunov and barrier functions. The level sets of Lyapunov functions naturally define invariant sets that can be used to certify safety. The work in \cite{tee2009barrier} used the notion of ``barrier Lyapunov function" to ensure stability under state constraints is achieved. As pointed out by \cite{ames2016control}, such conditions sometimes are overly strong and conservative. The more recent work \cite{romdlony2016stabilization} proposed the notion of (control) Lyapunov-barrier function (\ymmark{the lower-bounded function $W$ in \cite[Proposition 1 and Definition 2]{romdlony2016stabilization}}), and derived sufficient conditions for stability and stabilization with guaranteed safety. \ymmark{Despite the potential practical value of the control design framework presented in \cite{romdlony2016stabilization}, the type of Lyapunov-barrier functions considered in \cite{romdlony2016stabilization} (defined on $\Real^n$ and radially unbounded) implicitly imposes strong conditions on the unsafe set} (e.g., it has to be unbounded \cite[Theorem 11]{braun2017existence}). The authors of \cite{braun2017existence} then proposed sufficient conditions for safe stabilization using non-smooth control Lyapunov functions (see also \cite{braun2019complete}). {\color{black} The same authors also pointed out a technical inconsistency (see \cite{braun2020comment} for details) of the control Lyapunov-barrier conditions proposed in \cite{romdlony2016stabilization}. All these indicate that unifying Lyapunov and barrier functions is a non-trivial task.} 

In this paper, we approach this question of uniting Lyapunov and barrier functions from the converse direction. Different from the aforementioned work, we aim to formulate necessary (and sufficient) Lyapunov conditions for asymptotic stability under state constraints. We show that, if we restrict the domain of the Lyapunov function to the set of initial conditions from which solutions can simultaneously satisfy the conditions of asymptotic stability and safety, then a smooth Lyapunov function can be found, building upon earlier results on converse Lyapunov functions \cite{kurzweil1956inversion,teel2000smooth}. In particular, the results from  \cite{teel2000smooth} play a key role in inspiring us to formulate a Lyapunov function that is defined on the entire set of initial conditions from which the stability with safety specification is { satisfied}. We further extend the converse theorems to reach-avoid-stay type specifications, for which solutions of a system are required to reach a target set within a finite time and remain there after, while avoiding an unsafe set. Since reachability (similar to asymptotic attraction) does not impose any stability conditions (see Vinograd's example \cite[p.~120]{meiss2007differential}), we in general cannot expect to find a Lyapunov function that is defined \ymmark{in} a neighborhood of the target set. We use a robustness argument \cite{liu2020converse} to obtain a slightly weaker statement in the sense that if a reach-avoid-stay specification is satisfied robustly, then there exists a robust Lyapunov-barrier function that is robust under perturbations arbitrarily close to that of the original system.

The main contributions of the paper are summarized as follows. 
\begin{itemize}
    \item We formulate the problems of stability with safety and reach-avoid-stay specifications and establish connections between them.
    \item We prove a smooth converse Lyapunov-barrier function theorem that is defined on the entire set of initial conditions from which the stability with safety property is satisfied. 
    \item We extend the converse Lyapunov-barrier function theorem to reach-avoid-stay type specifications using a robustness argument. We show by example that such statements are the strongest one can obtain. 
    \item We extend the converse Lyapunov-barrier functions to converse control Lyapunov-barrier functions w.r.t. reach-avoid-stay specifications, \ymmark{provided that there exists a Lipschitz continuous feedback law.}
\end{itemize}

{\color{black} One clear limitation of the results of this paper is that the converse Lyapunov-barrier theorems are not constructive, as with classical converse Lyapunov theorems. Nonetheless, the unified necessary and sufficient conditions with a single Lyapunov-barrier function are of theoretical interest and can hopefully shed some light on developing computational approaches (see, e.g., \cite{Ravanbakhsh2017learning,berkenkamp2016safe,zhao2020synthesizing,ratschan2010providing,djaballah2017construction}) for stability with safety or reach-avoid-stay specifications.} 

The rest of this paper is organized as follows. In Section \ref{sec:prel}, we present the problem formulation. In Section \ref{sec:stability}, we prove a smooth converse Lyapunov-barrier function theorem for stability with safety guarantees. In Section \ref{sec:reach}, we extend the converse Lyapunov-barrier theorem to reach-avoid-stay type specifications. The paper is concluded in Section \ref{sec:con}. 

We list some notation used in this paper. We use $\Real^n$ to denote the Euclidean space of dimension $n>1$, $\Real$ the set of real numbers, $\Real_{>0}$ the set of positive real numbers, and $\Real_{\ge 0}$ the set of nonnegative real numbers. For $x\in\Real^n$ and $r\ge 0$, we denote the ball of radius $r$ centered at $x$ by $x+r\B=\set{y\in\Real^n:\,\abs{y-x}\le r}$, where $\abs{\cdot}$ is the Euclidean norm. For a closed set $A\subset\Real^n$ and $x\in\Real^n$, we denote the distance from $x$ to $A$ by $\norm{x}_{A}=\inf_{y\in A}\norm{x-y}$ and $r$-neighborhood of $A$ by $A+r\B=\cup_{x\in A}(x+r\B)=\set{x\in\Real^n:\,\norm{x}_A\le r}$. For a set $A\subseteq\Real^n$, $\overline{A}$ denotes its closure.  %
For two sets $A,B\in\Real^n$, we use $A\setminus B$ to denote the set difference defined by $A\setminus B=\set{x:\,x\in A,\,x\not\in B}$. 
We say a function $\alpha:\,\Real_{\ge 0}\ra\Real_{\ge 0}$ belongs to class $\mathcal{K}$ if it is continuous, zero at zero, and strictly increasing. It is said to belong to $\mathcal{K}_{\infty}$ if it belongs to class $\mathcal{K}$ and is unbounded. A function $\beta:\,\Real_{\ge 0}\times\Real_{\ge 0}\ra\Real_{\ge 0}$ is said to belong to class $\mathcal{KL}$ if, for each $t\ge 0$, $\beta(\cdot,t)$ belongs to class $\mathcal{K}$ and, for each $s\ge 0$, $\beta(s,\cdot)$ is decreasing and satisfies $\lim_{s\ra\infty}\beta(s,t)=0$.  

\section{Preliminaries} \label{sec:prel}

Consider a continuous-time dynamical system 
\begin{equation}\label{eq:sys}
     \dx= f(x),
\end{equation}
where $x\in\Real^n$ and $f:\,\Real^n\ra\Real^n$ is assumed to be locally Lipschitz. For each $x_0\in\Real^n$, we denote the unique solution starting from $x_0$ and defined on the maximal interval of existence by $\phi(t;x_0)$. For simplicity of notation, we may also write the solution as $\phi(t)$ if $x_0$ is not emphasized or as $\phi$ if the argument $t$ is not emphasized. %

Given a scalar $\delta\ge 0$, a \emph{$\delta$-perturbation} of the dynamical system (\ref{eq:sys}) is described by the differential inclusion
\begin{equation}\label{eq:p1}
    \dx \in F_{\delta}(x),
\end{equation}
where $F_{\delta}(x)=f(x)+\delta\B$. An equivalent description of the  $\delta$-perturbation of system (\ref{eq:sys}) can be given by %
\begin{equation}\label{eq:p2}
    \dx = f(x) + d,
\end{equation}
where $d:\,\Real\ra\delta\B$ is any measurable signal. We denote system (\ref{eq:sys}) by $\S$ and its $\delta$-perturbation by $\Sd$. Note that $\Sd$ reduces to $\S$ when $\delta=0$. A solution of $\Sd$ starting from $x_0$ can be denoted by $\phi(t;x_0,d)$, where $d$ is a given disturbance signal. We may also write the solution simply as $\phi(t)$ or $\phi$. 

We introduce some notation for reachable sets of $\Sd$. Denote the set of all solutions for $\Sd$ starting from $x_0$ by $\Phi_{\delta}(x_0)$.  Let $\R_{\delta}^t(x_0)$ denote the set reached by solutions of $\Sd$ at time $t$ starting from $x_0$, i.e.,
$$
\R_{\delta}^t(x_0) = \set{\phi(t):\,\phi\in \Pd(x_0))}. 
$$
For $T\ge 0$, we define 
$$
\R_{\delta}^{t\ge T}(x_0) = \bigcup_{t\ge T} \R_{\delta}^t(x_0),\quad  \R_{\delta}^{0\le t\le T}(x_0) = \bigcup_{0\le t\le T} \R_{\delta}^t(x_0),
$$
and write 
$$
\R_{\delta}(x_0) = \R_{\delta}^{t\ge 0}(x_0). 
$$
For a set $W\subseteq\Real^n$, 
\begin{align*}
\R_{\delta}^t(W)& =  \bigcup_{x_0\in W} \R_{\delta}^t(x_0),\\
\R_{\delta}^{t\ge T}(W) & =  \bigcup_{x_0\in W} \R_{\delta}^{t\ge T}(x_0), \\
\R_{\delta}^{0\le t\le T}(W) & =  \bigcup_{x_0\in W} \R_{\delta}^{0\le t\le T}(x_0), \\
\R_{\delta}(W) &=  \bigcup_{x_0\in W} \R_{\delta}(x_0). 
\end{align*}

\begin{defn}[Forward invariance]
A set $\Omega\subset\Real^n$ is said to be forward invariant for $\Sd$ (or $\delta$-robustly forward invariant for $\S$), if solutions from $\Omega$ are forward complete (i.e., defined for all positive time) $\mathcal{R}_\delta(\Omega)\subseteq\Omega$.
\end{defn}

\subsection{Reach-avoid-stay and stability with safety guarantees}

In this subsection, we formally define two common types of properties for solutions of $\Sd$ and highlight the connections between them. 

The first one is on reaching a target set in finite time and remaining there after, while avoiding an unsafe set. This is often called a reach-avoid-stay type specification. 

\begin{defn}[Reach-avoid-stay specification]
We say that $\Sd$ satisfies a reach-avoid-stay specification $(W,U,\Omega)$, where $W,U,\Omega\subseteq\Real^n$, if the following conditions hold:
\begin{enumerate}
    \item (reach and stay w.r.t. $\Omega$) Solutions of $\Sd$ starting from $W$ are defined for all positive time (i.e., forward complete) and there exists some $T\ge 0$ such that $\R_{\delta}^{t\ge T}(W)\subseteq \Omega$. 
    \item (safe w.r.t. $U$) $\R_{\delta}(W)\cap U=\emptyset$. 
\end{enumerate}
If these conditions hold, we also say that $\S$ $\delta$-robustly satisfies the reach-avoid-stay specification $(W,U,\Omega)$. 
\end{defn}

A closely related property for solutions of $\Sd$ is stability with safety guarantees. We first define stability for solutions of $\Sd$ w.r.t. a closed set. 

\begin{defn}[Set stability]\label{def:stability}
A closed set $A\subset\Real^n$ is said to be \textit{uniformly asymptotically stable} (UAS) for $\Sd$ if the following two conditions are met:
\begin{enumerate}
    \item (uniform stability) For every $\eps>0$, there exists a $\delta_{\eps}>0$ such that $\norm{\phi(0)}_A<\delta_\eps$ implies that $\phi(t)$ is defined for $t\ge 0$ and $\norm{\phi(t)}_A<\eps$ for any solution $\phi$ of $\Sd$ for all $t\geq 0$; and 
    \item (uniform attractivity) There exists some $\rho>0$ such that, for every $\eps>0$, there exists some $T>0$ such that $\phi(t)$ is defined for $t\ge 0$ and $\norm{\phi(t)}_A<\eps$ for any solution $\phi$ of $\Sd$ whenever $\norm{\phi(0)}_A<\rho$ and $t\ge T$. 
\end{enumerate}
If these conditions hold, we also say that $A$ is $\delta$-robustly UAS (or $\delta$-RUAS) for $\S$. 
\end{defn}

\begin{defn}[Domain of attraction]
If a closed set $A\subseteq\Real^n$ is $\delta$-RUAS for $\S$, we further define the domain of attraction of $A$ for $\Sd$, denoted by 
$\mathcal{G}_{\delta}(A)$, as the set of all initial states $x\in\Real^n$ such that any solution $\phi\in\Pd(x)$ is defined for all positive time and converges to the set $A$, i.e., 
\begin{equation*}
    \mathcal{G}_{\delta}(A) = \big\{x\in \Real^n:\,\forall \phi\in \Pd(x),   \lim_{t\ra\infty}\norm{\phi(t)}_{A}=0\big\}. 
\end{equation*}
\end{defn}

\begin{defn}[Stability with safety guarantee]
We say that $\Sd$ satisfies a stability with safety guarantee specification $(W,U,A)$, where $W,U,A\subseteq\Real^n$ and {$A$ is closed}, if the following conditions hold:
\begin{enumerate}
    \item (asymptotic stability w.r.t. $A$) The set $A$ is UAS for $\Sd$ and the domain of attraction of $A$ contains $W$, i.e. $W\subseteq\mathcal{G}_\delta(A)$. 
    \item (safe w.r.t. $U$) $\R_{\delta}(W)\cap U=\emptyset$. 
\end{enumerate}
If these conditions hold, we also say that $\S$ $\delta$-robustly satisfies the stability with safety guarantee specification $(W,U,A)$. 
\end{defn}

\section{Converse Lyapunov-Barrier Function for Stability with Safety Guarantees} \label{sec:stability}

In this section, we derive a converse Lyapunov-barrier function theorem for $\Sd$ satisfying a stability with safety guarantee specification $(W,U,A)$.

\begin{defn}\cite{teel2000smooth}
Let $A\subset\Real^n$ be a compact set contained in an open set $D\subset\Real^n$. A continuous function $\omega:\,D\ra\Real_{\ge 0}$ is said to be a proper indicator for $A$ on $D$ if the following two conditions hold: (1) $\omega(x)=0$ if and only if $x\in A$; (2) $\lim_{m\ra\infty} \omega(x_m)=\infty$ for any sequence $\set{x_m}$ in $D$ such that either $x_m\ra p\in \partial D$ or $\abs{x_m}\ra\infty$ as $m\ra\infty$. 
\end{defn} %

Intuitively, a proper indicator for a compact set $A\subset D$, where $D\subset\Real^n$ is open, is a continuous function whose value equals zero if and only if on $A$ and approaches infinity at the boundary of $D$ or at infinity. It generalizes the idea of a radially unbounded function. 

\begin{rem}
Let $A\subset\Real^n$ be a compact set contained in an open set $D\subset\Real^n$. There is always a proper indicator for $A$ on $D$ defined by \cite[Remark 2]{teel2000smooth}
$$
\omega(x) = \max\left\{ \norm{x}_A,  \frac{1}{\norm{x}_{\Real^n\setminus D}} - \frac{2}{\text{dist}(A,\Real^n\setminus D)} \right\},
$$
where $\text{dist}(A,\Real^n\setminus D)=\inf_{x\in A}\norm{x}_{\Real^n\setminus D}$. Indeed, $\omega$ is clearly continuous. If $x\in A$, we have $\omega(x)=\norm{x}_A=0$. If $x\in D\setminus{A}$, we have $\omega(x)\ge \norm{x}_A>0$. For any $\set{x_m}$ in $D$ such that either $x_m\ra p\in \partial D$ or $\abs{x_m}\ra\infty$ as $n\ra\infty$, we either have $\norm{x_m}_A\ra\infty$ or $\frac{1}{\norm{x_m}_{\Real^n\setminus D}} \ra\infty$. 
\end{rem}

\begin{thm}\label{thm:conv1}
Suppose that $A$ is compact, $U$ is closed, and $A\cap U=\emptyset$. Then the following two statements are equivalent: 
\begin{enumerate}
    \item $\Sd$ satisfies the stability with safety guarantee specification  $(W,U,A)$.
    \item There exists an open set $D$ such that $(A\cup W)\subset D$ and $D\cap U=\emptyset$, a smooth function $V:\,D\ra\Real_{\ge 0}$ and class $\mathcal{K}_{\infty}$ functions $\alpha_1$ and $\alpha_2$ such that, for all $x\in D$ and $d\in\delta \B$, \begin{equation}\label{eq:lyap1}
    \alpha_1(\omega(x)) \le V(x) \le \alpha_2(\omega(x)),
\end{equation}
and
\begin{equation}\label{eq:lyap2}
    \nabla V(x)\cdot(f(x)+d) \le -V(x),
\end{equation}
{where $\omega$ be any proper indicator for $A$ on $D$}, 
\end{enumerate}
Moreover, the set $D$ can be taken as the following set
\begin{align}
\mathcal{W}_{\delta} = \big\{x\in \Real^n:\,\forall \phi\in \Pd(x), &  \lim_{t\ra\infty}\norm{\phi(t)}_{A}=0 \text{ and }  \notag\\
& \phi(t)\not\in U,\forall t\ge 0\big\}. \label{eq:winset}
\end{align}
\end{thm}

Clearly, the set $\mathcal{W}_{\delta}$ defined above includes all initial states from which solutions of $\Sd$ will approach $A$ and avoid the unsafe set $U$. The following lemma establishes some basic properties of the set $\Wd$. The proof can be found in {Appendix A}. 

\begin{lem}\label{lem:winset}
Suppose that $A$ is compact, $U$ is closed, and $A\cap U=\emptyset$. If $\Sd$ satisfies a stability with safety guarantee specification $(W,U,A)$, then the set $\Wd$ is open, forward invariant, and satisfies $ W\subseteq \Wd\subseteq \Gd(A)$. 
\end{lem}

The proof of Theorem \ref{thm:conv1} relies on the following result, which states that, on any forward invariant open subset $D$ of $\Gd(A)$, we can find a ``global'' Lyapunov function relative to $D$. 

\begin{prop}\label{prop:clf}
Let $A\subset\Real^n$ be a compact set that is UAS for $\Sd$. Let $D\subset\Real^n$ be an open set such that $A\subset D\subset \Gd(A)$ and $D$ is forward invariant for $\Sd$, where  $\Gd(A)$ is the domain of attraction of $A$ for $\Sd$. Let $\omega$ be any proper indicator for $A$ on $D$. Then there exists a smooth function $V:\,D\ra\Real_{\ge 0}$ and class $\mathcal{K}_{\infty}$ functions $\alpha_1$ and $\alpha_2$ such that conditions (\ref{eq:lyap1}) and (\ref{eq:lyap2}) hold for all $x\in D$ and $d\in\delta \B$. 
\end{prop}

This proposition can be proved by combining the proof for Proposition 3 and the statements of Theorem 2 and Theorem 1 in \cite{teel2000smooth}. The main difference being that the results in \cite{teel2000smooth} are stated for more general differential inclusions and Proposition 3 in \cite{teel2000smooth} is proved on the \textit{whole domain of attraction} of $A$, whereas the above results are for specific $\delta$-perturbations of a Lipschitz ordinary differential equation and for any open forward invariant set containing the set $A$. {Due to this subtlety, Proposition 3 of \cite{teel2000smooth} is not directly applicable for our purpose.}
For completeness, we provide a more direct proof of this result in Appendix B.

\textbf{Proof of Theorem \ref{thm:conv1}} 

We first prove (2) $\Longrightarrow$ (1). The fact that $V$ is a smooth Lyapunov function, i.e., satisfying conditions (\ref{eq:lyap1}) and (\ref{eq:lyap2}), on an open neighborhood $D$ containing $A$ shows that $A$ is UAS for $\Sd$. We show that the set $D$ is forward invariant. Let $x_0\in D$. Then for any $\phi\in\Pd(x_0)$, we have 
$$
\frac{dV(\phi(t))}{dt}=\nabla V(\phi(t))\cdot (f(\phi(t)) + d(t))\le 0
$$
{holds for almost all $t\ge 0$.} It follows that $V(\phi(t))\le V(x_0)<\infty$. Hence $\phi(t)$ is bounded,  defined, and satisfies $\phi(t)\in D$ for all $t\ge 0$. By forward invariance of $D$ and $W\subset D$, we have $\Rd(W)\subset D$ and $\Rd(W)\cap U=\emptyset$. It remains to show that $W\subset \Gd(A)$. For any $x_0\in W$ and any $\phi\in\Pd(x_0)$, we have $\phi(t)\in D$ for all $t\ge 0$. Hence  
$$
\frac{dV(\phi(t))}{dt}=\nabla V(\phi(t))\cdot (f(\phi(t)) + d(t))\le - V(\phi(t))<0
$$
as long as $\phi(t)\not\in A$. A standard Lyapunov argument shows that $\norm{\phi(t)}_A\ra 0$ as $t\ra \infty$. Hence $x_0\in \Gd(A)$ and $W\subset \Gd(A)$. We have verified that $\Sd$ satisfies a stability with safety guarantee specification $(W,U,A)$. 

We then prove (1) $\Longrightarrow$ (2). By Lemma \ref{lem:winset}, we can let $D=\Wd$. Then $(A\cup W)\subset D\subset \Gd(A)$. Furthermore, $D$ is open and forward invariant. The conclusion follows from that of  Proposition \ref{prop:clf}. \hfill $\BlackBox$

\begin{rem}
Compared with related results on sufficient Lyapunov conditions for stability with safety guarantees (e.g.,  \cite{romdlony2016stabilization,braun2017existence,braun2019complete}), to the best knowledge of the authors, Theorem \ref{thm:conv1} provides the first converse Lyapunov-barrier theorem and we show that the converse Lyapunov function is defined on whole set of initial conditions from which asymptotic stability with safety guarantees is {satisfied}. In other words, we provide a Lyapunov characterization of the problem of asymptotic stability with safety guarantees. We also note that several converse barrier functions have been reported in the literature \cite{wisniewski2015converse,ratschan2018converse,liu2020converse}. In particular, the recent work \cite{liu2020converse} makes a connection between converse Lyapunov function and converse barrier function via a robustness argument, which, to some extent, inspired our work in this paper to unify converse Lyapunov and barrier functions. {\color{black} The results of this paper significantly differ from that in \cite{liu2020converse}, because converse results are established for both stability with safety guarantees and reach-avoid-stay specifications, whereas the results in \cite{liu2020converse} only concern safety. We achieved this non-trivial extension by adapting converse Lyapunov theorems (e.g., \cite{teel2000smooth}), as in Proposition \ref{prop:clf}, to work with safety requirements, enabled by characterizing all initial states from which solutions will satisfy stability with safety guarantees, as in  Lemma \ref{lem:winset}.} 
\end{rem}

While Theorem \ref{thm:conv1} gives a single smooth Lyapunov function satisfying the strong set of conditions (\ref{eq:lyap1}) and (\ref{eq:lyap2}), we propose the following set of sufficient conditions for two reasons. First, they appear to be weaker (although in fact theoretically equivalent in view of Theorem \ref{thm:conv1}) and perhaps easier to verify in practice {\cite{meng2021control}}. Second, they agree with the notions of Lyapunov and barrier functions commonly seen in the literature. 

\begin{prop}%
\label{prop:lyap-barrier}
Suppose that $A$ is compact, $U$ is closed, and $A\cap U=\emptyset$. If there exists an open set $D$ such that $(A\cup W)\subset D$ and smooth functions $V:\,D\ra\Real_{\ge 0}$ and $B:\,D\ra\Real$ such that 
\begin{enumerate}
    \item $V$ is positive definite on $D$ w.r.t. A, i.e., $V(x)=0$ if and only if $x\in A$; 
    \item $\nabla V(x)\cdot (f(x)+d)<0$ for all $x\in D\setminus A$ and $d\in\delta\B$; 
    \item $W\subset C=\set{x\in D:\,B(x)\ge 0}$ and $B(x)<0$ for all $x\in U$; 
    \item $\nabla B(x)\cdot (f(x)+d)\ge 0$ for all $x\in D$ and $d\in\delta\B$,
\end{enumerate}
then $\Sd$ satisfies the stability with safety guarantee specification  $(W,U,A)$. Furthermore, if $W$ is compact, then conditions (1)--(4) are also necessary for $\Sd$ to satisfy the stability with safety guarantee specification $(W,U,A)$. 
\end{prop}

\begin{proof}
We first prove the sufficiency part. Conditions (1)--(2) state that $V$ is a local Lyapunov function for $\Sd$ w.r.t. $A$. Hence $A$ is UAS for $\Sd$. Conditions (3)--(4) state that $B$ is a barrier function for $\Sd$ w.r.t. $(W,U)$. 

We can easily show that the set $C=\set{x\in D:\,B(x)\ge 0}$ is forward invariant. Indeed, if $C$ is not forward invariant, then there exists some $x_0\in C$, a solution $\phi\in\Pd(x_0)$, and some $\tau>0$ such that $B(\phi(\tau))<0$. Define 
$$
\overline{t}=\sup\{t\ge 0:\,\phi(t)\in C\}.
$$
Then $\overline{t}$ is well defined and finite. By continuity of $B(\phi(t))$, we have $B(\phi(\overline{t}))=0$. Since $\phi(\overline{t})\in D$ and $D$ is open, for $\eps>0$ sufficiently small, we have $\phi(t)\in D$ for \ymmark{almost all} $t\in [\overline{t},\overline{t}+\eps]$. This implies that, \ymmark{for almost all $ t\in[\overline{t},\overline{t}+\eps]$, }
$$
\frac{dB(\phi(t))}{dt}=\nabla B(\phi(t))\cdot (f(\phi(t)) + d(t))\ge 0.
$$ 
Hence we have $B(\phi(t))\ge B(\phi(\overline{t}))=0$ for \ymmark{almost all} $t\in [\overline{t},\overline{t}+\eps]$. This contradicts the definition of $\overline{t}$. Hence $C$ must be forward invariant. Since $W\subset C$ and $C\cap U=\emptyset$, we have $\Rd(W)\subset C$ and $\Rd(W)\cap U=\emptyset$. 

It remains to show that $W\subset \Gd(A)$. For any $x_0\in W$ and any $\phi\in\Pd(x_0)$, we have $\phi(t)\in C\subset D$ for \ymmark{almost} all $t\ge 0$. Hence  
$$
\frac{dV(\phi(t))}{dt}=\nabla V(\phi(t))\cdot (f(\phi(t)) + d(t))<0
$$
as long as $\phi(t)\not\in A$. A standard Lyapunov argument shows that $\norm{\phi(t)}_A\ra 0$ as $t\ra \infty$.

We then prove the necessity part. Since $A$ is compact, there exists a compact neighborhood $K$ of $A$ such that $A\subset K\subset D$. Let $c=\sup_{x\in K\cup W} V(x)$. Then $c>0$. Define $B(x)=c-V(x)$ for $x\in D$. We can easily verify that $V(x)$ and $B(x)$ satisfy conditions (1)--(4). 
\end{proof}

\begin{rem}\label{rem:ref29}
{\color{black} We compare the Lyapunov-barrier conditions with that in \cite{romdlony2016stabilization}, which provided a novel control framework for stabilization with guaranteed safety for nonlinear systems. Nonetheless, we restrict the formulation to autonomous systems (cf. Proposition 1 in \cite{romdlony2016stabilization}). This is without loss of generality, because the control framework in \cite{romdlony2016stabilization} is fundamentally built upon the conditions for autonomous systems, as clearly indicated in \cite{romdlony2016stabilization} (see, e.g., the remark before and proof of  \cite[Theorem 3]{romdlony2016stabilization}). We also change the notion slightly to be consistent with the notation used in this paper.}
In \cite{romdlony2016stabilization}, a set of sufficient conditions for a smooth function $V:\,\Real^n\ra {\mathbb{R}}$ to be called a Lyapunov-barrier function for the system (\ref{eq:sys}) with respect to the origin and an unsafe set $U$ were formulated as follows: 
\begin{enumerate}
    \item[(i)] $V$ is lower-bounded and radially unbounded;
    \item[(ii)] $V(x)>0$ for all $x\in U$;
    \item[(iii)] $\nabla V(x)\cdot f(x)<0$ for all $x\in \Real^n\setminus(U\cup \set{0})$; and 
    \item[(iv)] $\overline{\Real^n\setminus (U\cup C)} \cap \overline{U}=\emptyset$, where the set $C$ is given by $C=\set{x\in\Real^n:\,V(x)\le 0}$. 
\end{enumerate}
In \cite{braun2017existence}, it is shown that the above conditions imply the set $U$ is necessarily unbounded. Here we show another property that indicates the restrictive nature of condition (iv); that is, 
\begin{equation}\label{eq:iv}
x\in \partial U \text{ implies } V(x)=0. 
\end{equation}
In fact, suppose that this is not the case, then $V(x)>0$. There exists a sequence $\set{x_n}\ra x\in\partial D$ such that $V(x_n)>0$ (and hence $\set{x_n}\cap C=\emptyset$) and $\set{x_n}\cap U=\emptyset$ (this is possibly because $x\in\partial U$). Hence $\set{x_n}\subset\Real^n\setminus(U\cup C)$. It follows that $x\in \overline{\Real^n\setminus (U\cup C)}$. By condition (iv) above, $x\not\in \overline{U}$, which contradicts $x\in\partial U$. In view of (\ref{eq:iv}), condition (iv) above is somewhat restrictive, because it implies that the boundary of the unsafe set $U$ lies entirely on a level curve of $V$. %
\end{rem}

\begin{rem}
Figure \ref{fig:sets} provides an illustration of the sets defined for proving Theorem \ref{thm:conv1}.  
\begin{figure}[h]
\centering
\includegraphics[width=0.5\textwidth]{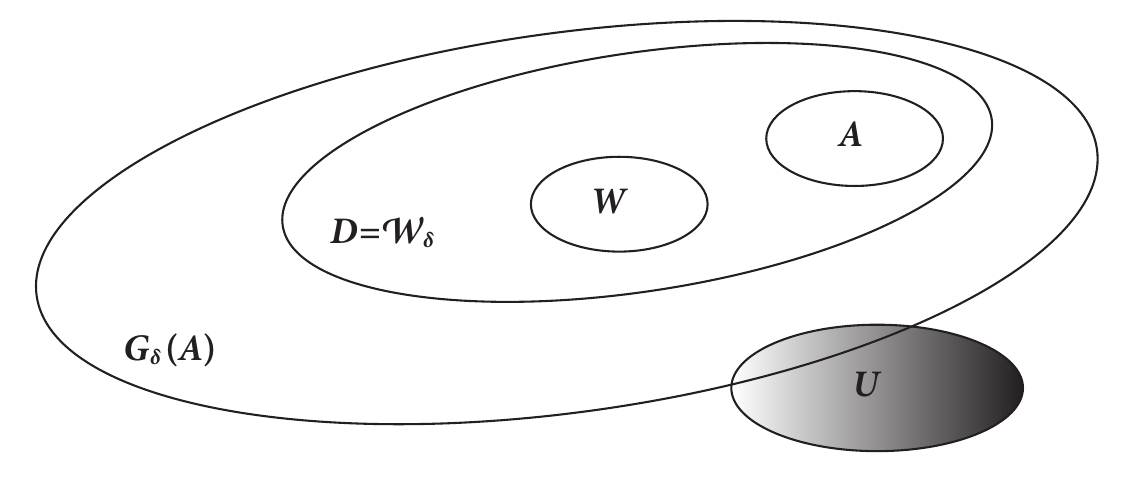}
\caption{An illustration of the sets involved in Theorem \ref{thm:conv1}, Lemma \ref{lem:winset}, and Proposition \ref{prop:clf}. While the domain of attraction $\Gd(A)$ can potentially intersect with the unsafe set $U$, the winning set $\Wd$ defined in (\ref{eq:winset}) characterizes the set of initial conditions from which the stability with safety constraints is satisfied. Clearly, the system $\Sd$ satisfies a stability with safety specification $(W,U,A)$ if and only if $W\subset \Wd$. Theorem \ref{thm:conv1} (together with Lemma \ref{lem:winset} and Proposition \ref{prop:clf}) states that a smooth Lyapunov function can be found on the set $D=\Wd$ to verify the specification $(W,U,A)$. 
}\label{fig:sets}
\end{figure}
\end{rem}

\section{Converse Lyapunov-Barrier Function for Reach-Avoid-Stay Specifications} \label{sec:reach}

The converse results proved in the previous section can be extended to reach-avoid-stay specifications under some mild modifications. 

Suppose that $\Sd$ satisfies a reach-avoid-stay specification $(W,U,\Omega)$. 

\begin{lem}\label{lem:A}
Suppose that $\Omega$ is compact and $W$ is nonempty. If $\Sd$ satisfies a reach-avoid-stay specification $(W,U,\Omega)$, then the set 
\begin{equation}\label{eq:A}
A = \set{x\in\Omega:\,\forall\phi\in\Pd(x),\phi(t)\in\Omega,\forall{t\ge 0}}. 
\end{equation}
is a nonempty compact invariant set for $\Sd$. 
\end{lem}

\begin{proof}
We first show that $A$ is nonempty. By the definition of reach-avoid-stay specification $(W,U,\Omega)$, solutions of $\Sd$ starting from $W$ are forward complete and there exists some $T\ge 0$ such that $\Rd^{t\ge T}(W)\subset\Omega$. It is easy to verify that the set $\Rd^{t\ge T}(W)$ is forward invariant for $\Sd$. Clearly, $\Rd^{t\ge T}(W)\subset A$ and $A$ is nonempty. 

We next show that $A$ is compact. Since $A\subset\Omega$ and $\Omega$ is compact, we only need to show that $A$ is closed. Note that $A$ is forward invariant by definition. Let $\set{x_m}$ be a sequence in $A$ that converges to $x$. Since $\Omega$ is compact, we have $x\in \Omega$. Suppose that $x\not\in A$. Then there exists some $\phi\in\Pd(x)$ and some $\tau>0$ such that $\phi(\tau)\not\in\Omega$. By continuous dependence of solutions of $\Sd$ on initial conditions, there exists a sequence of solutions $\phi_m\in\Pd(x_m)$ that converges to $\phi$ uniformly on $[0,\tau]$. We have $\phi_m(\tau)\ra\phi(\tau)\not\in\Omega$ as $m\ra\infty$. Since $\Real^n\setminus\Omega$ is open, this implies that for $m$ sufficiently large, $\phi_m(\tau)\not\in\Omega$. This contradicts the definition of $A$ (recall that $x_m\in A$ and $\phi_m\in\Pd(x_m)$). Hence $x\in A$ and $A$ is compact. %
\end{proof}

The following proposition states that any compact robustly invariant set of $\Sd$ is {UAS} for $\S_{\delta'}$, where $\delta'$ can be taken to be arbitrarily close to $\delta$. This fact was essentially proved in \cite{liu2020converse} in a slightly different context. The conclusion does not hold for $\delta'=\delta$ (see Example \ref{ex:uas}). 

\begin{prop}\label{prop:uas}
Any nonempty compact invariant set $A$ for $\Sd$ is UAS for $\S_{\delta'}$ whenever $\delta'\in[0,\delta)$. 
\end{prop}

The proof relies on the following technical lemma from \cite{liu2020converse}. 

\begin{lem}\cite{liu2020converse}\label{lem:control}
Fix any $\delta'\in(0,\delta)$ and $\tau>0$. Let $K\subset\Real^n$ be a compact set. Then there exists some $r=r(K,\tau,\delta',\delta)>0$ such that the following holds: if there is a solution $\phi$ of $\S_{\delta'}$ such that $\phi(s)\in K$ for all $s\in[0,T]$, where $T\ge\tau$, then for any $y_0\in \phi(0)+r\B$ and any $y_1\in \phi(T)+r\B$, we have $y_1\in \R_{\delta}^T(y_0)$, i.e., $y_1$ is reachable at $T$ from $y_0$ by a solution of $\Sd$. 
\end{lem}

We present the proof of Proposition \ref{prop:uas} as follows.

\begin{proof}
We verify conditions (1) uniform stability and (2) uniform attractivity as required by Definition \ref{def:stability}. 

(1) For any $\eps>0$, let $\tau>0$ be the minimal time that is required for solutions of $\S_{\delta'}$ to \ymmark{travels} from \ymmark{interior of} $A+\frac{\eps}{2}\B$ to \ymmark{$\Real\setminus (A+\eps\B)$}. The existence of such a $\tau$ follows from that $f$ is locally Lipschitz and an argument using Gronwall's inequality. Pick $\delta_\eps<\min(r,\frac{\eps}{2})$, where $r$ is from Lemma \ref{lem:control}, applied to the set $A+\eps\B$ and scalars $\tau$, $\delta'$, and $\delta$. Let $\phi$ be any solution of $\S_{\delta'}$ such that $\norm{\phi(0)}_A<\delta_\eps$. We show that $\norm{\phi(t)}_A<\eps$ for all $t\ge 0$. Suppose that this is not the case. Then $\norm{\phi(t_1)}_A\ge\eps$ for some $t_1\ge\tau>0$. Since $\delta_\eps<r$ and $A$ is compact, we can always pick $y_0\in A$ such that $y_0\in \phi(0)+r\B$. By Lemma \ref{lem:control}, there exists a solution of $\Sd$ from $y_0\in A$ to $y_1=\phi(t_1)\not\in A$. This contradicts that $A$ is forward invariant for $\Sd$. 

(2) Fix any $\eps_0>0$. Following part (1), choose $\delta_{\eps_0}$ such that $\norm{\phi(0)}_A<\delta_{\eps_0}$ implies $\norm{\phi(t)}_A<\eps_0$ for any solution $\phi(t)$ of $\S_{\delta'}$. Let $r$ be chosen according to Lemma \ref{lem:control} with the set $A+\eps_0\B$ and scalars $\tau=1$, $\delta'$, and $\delta$. Choose $\rho\in(0,r)$. Let $\phi$ be any solution of $\S_{\delta'}$. We show that $\norm{\phi(0)}_A<\rho$ implies that $\phi(t)\in A$ for all $t\ge 1$. Suppose that this is not the case. Then there exists some $t_1\ge 1$ such that $\phi(t_1)\not\in A$. Since $\rho<r$, we can pick $y_0$ such that $y_0\in \phi(0)+r\B$ and $y_0\in  A$. By Lemma \ref{lem:control}, there exists a solution of $\Sd$ from $y_0\in A$ to $y_1=\phi(t_1)\not\in A$. This contradicts that $A$ is forward invariant for $\Sd$. Hence $\phi(t)\in A$ for all $t\ge 1$. 
This clearly implies (2). 
\end{proof}

Proposition \ref{prop:uas} establishes a link between robust invariance and asymptotic stability. By combining Lemma \ref{lem:A}, Proposition \ref{prop:uas}, and Theorem \ref{thm:conv1}, we can obtain the following converse theorem for a reach-avoid-stay specification. 

\begin{thm}\label{thm:conv2}
Suppose that $\Omega$ is compact, $U$ is closed, and $\Omega\cap U=\emptyset$, and $\Sd$ satisfies the  reach-avoid-stay specification $(W,U,\Omega)$. Then there exists a compact set $A\subset\Omega$ such that, for any $\delta'\in[0,\delta)$ and any proper indicator $\omega$ for $A$ on $D$, there exists an open set $D$ such that $(A\cup W)\subset D$ and $D\cap U=\emptyset$, a smooth function $V:\,D\ra\Real_{\ge 0}$ and class $\mathcal{K}_{\infty}$ functions $\alpha_1$ and $\alpha_2$ such that conditions (\ref{eq:lyap1}) and (\ref{eq:lyap2}) hold for all $x\in D$ and $d\in\delta'\B$. 
\end{thm}

\begin{proof}
By Lemma \ref{lem:A}, there exists a compact set $A\subset \Omega$ that is $\delta'$-UAS for any $\delta'\in[0,\delta)$ by Proposition \ref{prop:uas}. Furthermore, as shown in the proof of Lemma \ref{lem:A}, $\Rd^{t\ge T}(W)\subset A$. This implies that, for any $\delta'\in[0,\delta)$, the domain of attraction of $A$ for $\S_{\delta'}$ includes $W$. Hence $\S_{\delta'}$ satisfy the  stability with safety guarantee specification $(W,U,A)$. The conclusion follows from that of Theorem \ref{thm:conv1}.
\end{proof}

\begin{rem}
Figure \ref{fig:sets2} provides an illustration of the sets defined for proving Theorem \ref{thm:conv2}.  
\begin{figure}[h]
\centering
\includegraphics[width=0.5\textwidth]{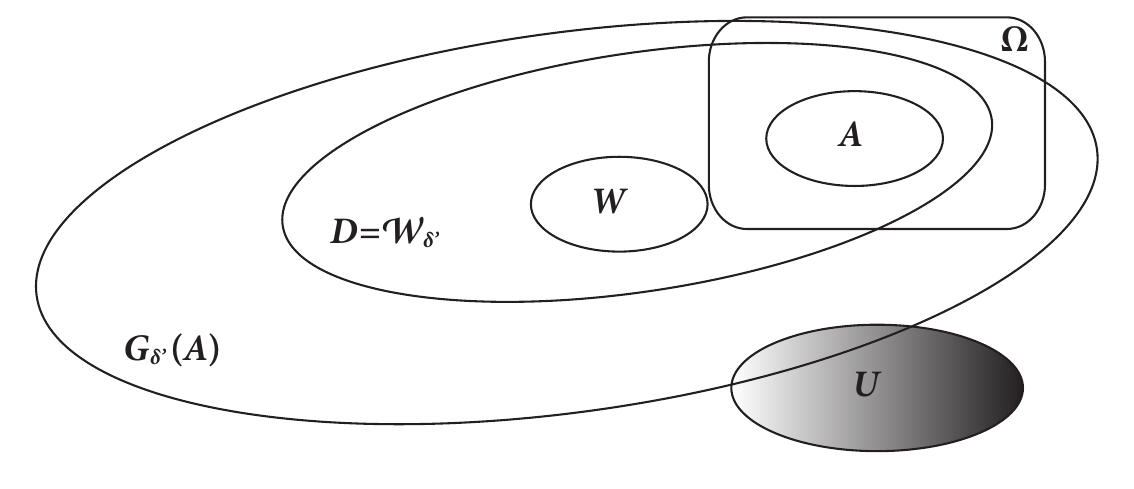}
\caption{An illustration of the sets involved in Theorem \ref{thm:conv2}. If a reach-avoid-stay specification $(W,U,\Omega)$ is satisfied, then for each $\delta'\in[0,\delta)$, we can find a set $A$ such that $\S_{\delta'}$ satisfies the stability with safety guarantee specification $(W,U,A)$. Consequently, a set $D$ and a Lyapunov function $V$ defined on $D$ can be found such that the Lyapunov conditions (\ref{eq:lyap1}) and (\ref{eq:lyap2}) hold for $\S_{\delta'}$. The conclusion of Theorem \ref{thm:conv2} follows from that of Theorem \ref{thm:conv1}.}
\label{fig:sets2}
\end{figure}
\end{rem}

It would be tempting to draw a stronger conclusion than the one in Theorem \ref{thm:conv2} by allowing $\delta'=\delta$. The following example shows that the conclusion of Theorem \ref{thm:conv2} cannot be strengthened in this regard:  Under the current assumptions of Theorem \ref{thm:conv2}, there may not exist a converse Lyapunov-barrier function satisfying conditions (\ref{eq:lyap1}) and (\ref{eq:lyap2}) for $\Sd$, even if $\Sd$ satisfies a reach-avoid-stay specification $(W,U,\Omega)$. 

\begin{exmp}\label{ex:uas}
Consider $\S$ defined by \ymmark{$\dx = - x + x^2$}. Let $W=[-1,-0.9]$, $U=[0.6,\infty)$, $\Omega=[-0.25,0.5]$, and $\delta=0.25$. It is easy to verify that $\Sd$ satisfies the reach-avoid-stay specification $(W,U,\Omega)$. However, solutions of $\Sd$ starting from $x_0=0.5+\eps$, where $\eps>0$, with $d(t)=\delta$ will tend to infinity. Furthermore, for any $x_0\in\Omega$, there exists a solution of $\Sd$ that approaches 0.5. Hence, there does not exist an open set $D$ as in Theorem \ref{thm:conv2} and a converse Lyapunov-barrier function defined on $D$ that satisfies conditions (\ref{eq:lyap1}) and (\ref{eq:lyap2}) for all $x\in D$ and $d\in\delta\B$. The reason for this is that the conclusion of Proposition \ref{prop:uas} does not hold for $\Sd$, i.e., the set $A$ defined by (\ref{eq:A}) may not be UAS for $\Sd$, even though it is UAS for $\S_{\delta'}$ whenever $\delta'\in[0,\delta)$. Indeed, it is not difficult to verify that the set $A=[\frac{1}{2}-\frac{\sqrt{2}}{2},0.5]$ and, by the observation above, the set $A$ is not UAS for $\Sd$. 
\end{exmp}

Similarly, Proposition \ref{prop:lyap-barrier} can be adapted to give the following version of converse theorem for reach-avoid-stay specifications. 

\begin{prop}%
\label{prop:lyap-barrier2}
Suppose that $\Omega$ and $W$ are compact, $U$ is closed, and $\Omega\cap U=\emptyset$, and $\Sd$ satisfies the reach-avoid-stay specification $(W,U,\Omega)$. Then for any $\delta'\in[0,\delta)$, there exists a compact $A\subset\Omega$, an open set $D$ such that $(A\cup W)\subset D$,  and smooth functions $V:\,D\ra\Real_{\ge 0}$ and $B:\,D\ra\Real$ such that 
\begin{enumerate}
    \item $V$ is positive definite on $D$ w.r.t. A, i.e., $V(x)=0$ if and only if $x\in A$; 
    \item $\nabla V\cdot (f(x)+d)<0$ for all $x\in D\setminus A$ and $d\in\delta'\B$; 
    \item $W\subset C=\set{x\in D:\,B(x)\ge 0}$ and $B(x)<0$ for all $x\in U$; 
    \item $\nabla B\cdot (f(x)+d)\ge 0$ for all $x\in D$ and $d\in\delta'\B$.
\end{enumerate}
\end{prop}

\begin{proof}
Similar to that of Proposition \ref{prop:lyap-barrier}. 
\end{proof}
The above converse results (Theorem \ref{thm:conv2} and Proposition \ref{prop:lyap-barrier2}) reveal that the verification and design for reach-avoid-stay specifications can indeed be centered around the problem of stability/stabilization with safety guarantees. This is \textit{without loss of generality} at least from a robustness point of view. In this regard, Lemma \ref{lem:A} and Proposition \ref{prop:uas} connect robust reach-avoid-stay specification with stability with safety guarantees. We can also prove a result in the converse direction. These statements are summarized in the following proposition.

\begin{prop}
\begin{enumerate}
    \item If $\Sd$ satisfies a stability with safety guarantee specification $(W,U,A)$ and $W$ is compact, then for every $\eps>0$, $\Sd$ satisfies the reach-avoid-stay specification $(W,U,A+\eps\B)$.
    \item If $\Sd$ satisfies a reach-avoid-stay specification $(W,U,\Omega)$, then there exists a compact set $A\subseteq\Omega$ such that, for any $\delta'\in[0,\delta)$, $\S_{\delta'}$ satisfies the stability with safety guarantee specification $(W,U,A)$. 
\end{enumerate}
\end{prop}

\begin{proof}
(1) The conclusion follows from the uniform attractivity property for solutions for $\Sd$ under the stability assumption (Proposition 29 in Appendix A. \\
(2) It follows from Lemma \ref{lem:A}, Proposition \ref{prop:uas}, and the definitions of the specifications.  
\end{proof}

\section{{Converse Control Lyapunov-Barrier Functions for Reach-Avoid-Stay Specifications}} \label{sec:control}

{

In this section, we take advantage of the results from Section \ref{sec:reach} and make a straightforward derivation on a converse control Lyapunov-barrier function theorem for $\Sd$  satisfying a  reach-avoid-stay specification $(W,U,\Omega)$ under controls. We first recast the notion from Section \ref{sec:prel} for control systems.

Given a nonempty compact convex set of control inputs  $\U\subset\mathbb{R}^p$, 
consider a nonlinear system of the form
\begin{equation}\label{eq:sys3}
\dot{x}=f(x)+g(x)u+d, 
\end{equation}
where the mapping $g:\mathbb{R}^n\rightarrow\Real^{n\times p}$ is smooth;  $u:\Real_{\geq 0}\rightarrow \U$ is a locally bounded measurable control signal, whilst the other notation remains the same.

\begin{defn}[Control strategy]
A control strategy is a function
\begin{equation}
   \kappa:\Real^n\ra \U. 
\end{equation}
\end{defn}
\iffalse
\begin{defn}[State-dependent control]
We say that a
control signal u conforms to a control strategy $\kappa$ for \eqref{eq:sys3}, denoted by $u\in\ub_\kappa$, if 
\begin{equation}
    u(t)\in\kappa(x(t)),\quad\forall t\geq 0.
\end{equation}
The set of all control
signals that confirm to $\kappa$ is denoted by $\ub_\kappa$.
\end{defn}

\fi
We further denote $\Sd^\kappa$ by the control system driven by \eqref{eq:sys3} that is comprised by $u=\kappa(x)$.
\begin{defn}[Reach-avoid-stay controllable]
A system $\Sd$ is called reach-avoid-stay controllable w.r.t.  $(W,U,\Omega)$, where $W,U,\Omega\subseteq \Real^n$, if there exists a Lipschitz continuous control strategy $\kappa$ such that the system $\Sd^\kappa$ satisfies the reach-avoid-stay specification $(W,U,\Omega)$.
\end{defn}

Now we are ready to show that reach-avoid-stay controllability implies the existence of a control Lyapunov-barrier function w.r.t. the reach-avoid-stay specification. 

\begin{thm}
Suppose that $\Omega$ is compact, $U$ is closed, and $\Omega\cap U=\emptyset$, and $\Sd$ is reach-avoid-stay controllable w.r.t. $(W,U,\Omega)$. Then there exists a compact set $A\subset\Omega$ such that, for any $\delta'\in[0,\delta)$ and any proper indicator $\omega$ for $A$ on $D$, there exists an open set $D$ such that $(A\cup W)\subset D$ and $D\cap U=\emptyset$, a smooth function $V:\,D\ra\Real_{\ge 0}$ and class $\mathcal{K}_{\infty}$ functions $\alpha_1$ and $\alpha_2$ such that, for all $x\in D$ and $d\in\delta'\B$, equation \eqref{eq:lyap1} is satisfied and 
\begin{equation}
    \inf\limits_{u\in \mathcal{U}}\sup\limits_{x\in D}\sup\limits_{d\in\delta\mathcal{B}}[L_fV(x,d)+L_gV(x)u+V(x)]\leq  0.
\end{equation}
\end{thm}

\begin{proof}
By assumption, there exists a Lipschitz continuous $\kappa$ that renders the
solutions satisfy reach-avoid-stay specification $(W,U,\Omega)$. 
\iffalse
By \cite[Proposition 2.22]{freeman2008robust}, there exists a locally Lipschitz function $\phi: D\times \U\ra\U$ in $x$ such that 
$$\kappa(x)=\phi(x,\B_\U) $$
for all $(x,u)\in D\times \U$, where $\B_\U$ is the closed unit ball in $\U$.
\fi
Then by Theorem \ref{thm:conv2}, for any proper indicator $\omega$ for $A$ on $D$, there exists a function $V: D\rightarrow\Real_{\geq 0}$ satisfying \eqref{eq:lyap1} and
$$\sup\limits_{d\in\delta'\mathcal{B}} [L_fV(x,d)+L_gV(x)\kappa(x)+V(x)]\leq 0  $$
for all $x\in D$.  Taking the supremum over all $x\in D$, we have 
$$\sup\limits_{x\in D}\sup\limits_{d\in\delta'\mathcal{B}} [L_fV(x,d)+L_gV(x)\kappa(x)+V(x)]\leq 0.   $$ 
Since we have the control  $\kappa(x)\in\U$, it follows that
$$\inf\limits_{u\in \U}\sup\limits_{x\in D}\sup\limits_{d\in\delta'\mathcal{B}} [L_fV(x,d)+L_gV(x)u+V(x)]\leq 0. $$ 
\end{proof}
With a similar approach, Proposition \ref{prop:lyap-barrier2} can be applied to give the following version of converse control Lyapunov-barrier functions theorem for reach-avoid-stay specifications. 

\begin{prop}%
\label{prop:lyap-barrier3}
Suppose that $\Omega$ and $W$ are compact, $U$ is closed, and $\Omega\cap U=\emptyset$, and $\Sd$ is reach-avoid-stay controllable w.r.t. $(W,U,\Omega)$. Then for any $\delta'\in[0,\delta)$, there exists a compact $A\subset\Omega$, an open set $D$ such that $(A\cup W)\subset D$,  and smooth functions $V:\,D\ra\Real_{\ge 0}$ and $B:\,D\ra\Real$ such that 
\begin{enumerate}
    \item $V$ is positive definite on $D$ w.r.t. A, i.e., $V(x)=0$ if and only if $x\in A$; 
    \item 
    $\inf\limits_{u\in \U}\sup\limits_{x\in D}\sup\limits_{d\in\delta'\mathcal{B}} [L_fV(x,d)+L_gV(x)u]\leq 0$;
    \item $W\subset C=\set{x\in D:\,B(x)\ge 0}$ and $B(x)<0$ for all $x\in U$; 
    \item 
    $\sup\limits_{u\in \U} [L_fB(x,d)+L_gB(x)u]\geq 0$
for all $x\in D$ and $d\in\delta'\B$.
\end{enumerate}
\end{prop}

}

\section{Conclusions}\label{sec:con}

In this paper, we proved two converse Lyapunov-barrier function theorems for nonlinear systems satisfying either asymptotic stability with a safety constraints or a reach-avoid-stay type specification. In the former case, we show that a smooth Lyapunov-barrier function can be defined on the entire set of initial conditions from which asymptotic stability with a safety constraint can be { satisfied}. For the latter, we establish a converse theorem via a robustness argument. It is shown by example that the statement cannot be strengthened without additional assumptions. { We further extend the results to establish converse control Lyapunov-barrier functions for systems with control inputs.}

The focus of the current paper is on converse Lyapunov-barrier functions, applying which we make a quick extension to converse control Lyapunov-barrier function. %
{\color{black} There are two limitations in our work. We only considered an additive measurable disturbance in the right-hand side of the dynamical systems for the purpose of establishing converse Lyapunov-barrier results. In addition, similar to other converse Lyapunov theorems, the existence results are not constructive.}

An interesting future direction is to explore computational techniques for constructing Lyapunov-barrier function that is defined on the whole set of initial conditions (or as large a subset as possible of this set) from which a stability with safety guarantee or reach-avoid-stay specification is achievable, {for instance, learning techniques \cite{Ravanbakhsh2017learning,berkenkamp2016safe,zhao2020synthesizing} or interval analysis \cite{ratschan2010providing,djaballah2017construction}}. In this regard, the results of this paper (especially Theorems \ref{thm:conv1} and \ref{thm:conv2}) can hopefully shed some light into the development of such computational techniques with completeness (or approximate completeness) guarantees. 

\begin{ack}                               %
This work was supported in part by the Natural Sciences and Engineering Research Council of Canada, the Canada Research Chairs program, and an Early Researcher Award from the Ontario Ministry of Research, Innovation and Science.   %
\end{ack}

\bibliographystyle{plain}        %
\bibliography{clb} 

\appendix

\section{Proof of Lemma \ref{lem:winset}}\label{sec:open}

We first state two lemmas on the properties of the solutions of $\Sd$. 

The first one is well known from the basic theory of ODEs (see, e.g, \cite[Theorem 55, Appendix C]{sontag1998mathematical}). 

\begin{lem}[Continuous dependence]
Suppose that for some $x_0\in\Real^n$ there exists some $T>0$ such that solutions for $\Sd$ starting from $x_0$ are defined on $[0,T]$. Then there exists some $\delta>0$ such that solutions starting from $x_0+\delta\B$ are also defined on $[0,T]$ and there exists a constant $C$ (depending on $T$ and $x_0$) such that 
$$
\abs{\phi(t;x,d) - \phi(t;x_0,d)}\le C\abs{x-x_0}
$$
for all $x\in x_0+\delta\B$ and $d:\,[0,T]\ra \delta\B$.
\end{lem}

The next result is on topological properties of solutions of differential inclusions satisfying some basic conditions. It can be found, e.g., in \cite[Theorem 3, Section 7]{filippov1988differential}. Note that the differential inclusion we consider $\Sd:\,  x'\in F_{\delta}(x):= f(x)+\delta\B$ straightforwardly satisfies the basic conditions there (i.e., $F_\delta$ is upper semicontinuous and takes nonempty, compact, and convex values).

\begin{lem}[Compactness of reachable sets]\label{lem:compact}
Let $K\subset\Real^n$ be a compact set. Suppose that there exists some $\tau>0$ such that solutions of $\Sd$ starting from $K$ are always defined on $[0,\tau)$. Then, for any $T\in[0,\tau)$, $\Rd^{0\le t\le T}(K)$ is a compact set. Furthermore, solutions of $\Sd$ on $[0,T]$ form a compact set under the uniform convergence topology. 
\end{lem}

The following result shows that under the uniform stability assumption (i.e., condition (1) in Definition \ref{def:stability}), attraction of solutions starting from any compact set within the domain of attraction is always uniform. The proof of the following result is modeled after the proof for Proposition 3 in \cite[cf. Claim 4]{teel2000smooth}. 

\begin{prop}[Uniformity of attraction]\label{prop:uniform}
Suppose that a closed set %
$A\subset\Real^n$ is uniformly stable for $\Sd$, i.e., condition (1) of Definition \ref{def:stability} holds. Let $K$ be a compact set. Then the following two statements are equivalent:
\begin{enumerate}
    \item For any $x_0\in K$ and any $\phi\in\Phi_{\delta}(x_0)$, $\phi$ is defined for all $t\ge 0$ and 
    $$
    \lim_{t\ra\infty}\norm{\phi(t)}_A = 0.
    $$
    \item For every $\eps>0$, there exists $T=T(\eps)>0$ such that 
    $$
    \norm{\phi(t)}_A < \eps
    $$
    holds for any $x_0\in K$, $\phi\in\Pd(x_0)$, and $t\ge T$. 
    \end{enumerate}
\end{prop}

\begin{proof}
Clearly, (2) implies (1). We prove that (1) also implies (2) under the uniform stability assumption. Suppose that (2) does not hold. Then there exists some $\eps_0>0$ such that for all $n>0$ there exists $x_n\in K$,  $\phi_n\in \Pd(x_n)$, and $t_n\ge n$ such that 
\begin{equation}\label{eq:phi0}
\norm{\phi_n(t_n)}_A\ge \eps_0. 
\end{equation}
Let $\delta_0=\delta_{\eps_0}$ be given by condition (1) of Definition \ref{def:stability}. For every $n>0$, we must have
\begin{equation}
    \norm{\phi_n(t)}_A\ge \delta_0,\quad \forall t\in[0,n].
\end{equation}

\begin{claim}\label{clm:conv}
There exist subsequences $\set{x_n}$ and $\phi_n\in \Pd(x_n)$ such that $x_n$ converges to $x$ and $\phi_n$ converges to a solution $\phi\in \Pd(x)$. The latter convergence is uniform on every compact interval of $\Real_{\ge 0}$. 
\end{claim}

\textbf{Proof of Claim \ref{clm:conv}} From (1), we know that solutions starting from $K$ are always forward complete. Since $K$ is compact, we can assume without loss of generality that $\set{x_n}$ converges to $x\in K$ (otherwise we can pick a subsequence). By Lemma \ref{lem:compact}, there exists a subsequence of $\set{\phi_n}$, denoted by $\set{\phi_{1m}}$, that converges uniformly on $[0,1]$ to a solution $\phi_1\in \Pd(x)$. By the same argument, $\set{\phi_{1m}}$ has a subsequence, denoted by $\set{\phi_{2m}}$, that converges uniformly on $[0,2]$ to a solution $\phi_2\in \Pd(x)$. Repeat  this argument and pick the diagonal $\set{\phi_{mm}}$. Then  $\set{\phi_{mm}}$ has the claimed property. \qed

Let $\phi\in\Pd(x)$ be given by the claim. By statement (1), there exists $T>0$ such that 
\begin{equation}\label{eq:phi1}
    \norm{\phi(t)}_A<\frac{\delta_0}{2},\quad \forall t\ge T.
\end{equation}
However, since $\set{\phi_{mm}}$ converges to $\phi$ uniformly on $[0,T]$, there exists some $n\ge T$ such that 
\begin{equation}\label{eq:phi2}
    \abs{\phi_{n}(t)-\phi(t)}<\frac{\delta_0}{2},\quad \forall t\in [0,T].
\end{equation}
The equations (\ref{eq:phi1}) and (\ref{eq:phi2}) give $\norm{\phi_n(T)}_A<\delta_0$, which contradicts (\ref{eq:phi0}). 
\end{proof}

\textbf{Proof of Lemma 9} We can easily verify that $\Wd$ is forward invariant and $W\subset\Wd\subseteq \Gd(A)$ by its definition. We show that $\Wd$ is open. 

Let $x_0\in \Wd$. Let $\rho>0$ be given by condition (2) from Definition \ref{def:stability} for UAS of $A$. Choose $\eps_0<\rho$ such that 
$
(A+\eps_0\B)\cap U=\emptyset.
$
Choose $\delta_0=\delta_{\eps_0}$ according to condition (1) in Definition \ref{def:stability} for UAS of $A$. Clearly, $\delta_0\le \eps_0<\rho$. 

Then, by Proposition \ref{prop:uniform} in the Appendix, there exists some $T=T(\delta_0)>0$ such that $$\norm{\phi(t)}_A<\frac{\delta_0}{2}$$ for any solution $\phi\in\Pd(x_0)$ and all $t\ge T$. By Lemma \ref{lem:compact} in the Appendix, the set $K=\Rd^{0\le t\le T}(x_0)$ is compact. Let $\eps_1<\frac{\eps_0}{2}$ be chosen such that $(K+\eps_1\B)\cap U=\emptyset$. 

By continuous dependence of solutions of $\Sd$ with respect to initial conditions, there exists some $r>0$ such that, for all $x\in x_0+r\B$ and any $\psi\in \Pd(x)$, there exists a solution $\phi\in \Pd(x_0)$ such that 
$$
\abs{\phi(t)-\psi(t)}<\eps_1,\quad \forall t\in [0,T]. 
$$
It follows that 
\begin{equation}\label{eq:safe1}
\Rd^{\le T}(x_0+\delta\B)\subset K+\eps_1\B.
\end{equation}
Furthermore, at $t=T$, we have $\norm{\psi(T)}_A\le \norm{\phi(T)}_A+\eps_1<\frac{\delta_0}{2}+\frac{\delta_0}{2}=\delta_0$. It follows from condition (1) in Definition \ref{def:stability} that 
\begin{equation}\label{eq:safe2}
\psi(t)\in A+\eps_0\B \subset A+\rho\B
\end{equation}
for all $\psi\in \Pd(x)$, $x\in x_0+r\B$, and $t\ge T$. By condition (2) in Definition \ref{def:stability}, $\lim_{t\ra\infty}\norm{\psi(t)}_A=0$. In view of (\ref{eq:safe1}) and (\ref{eq:safe2}), $\psi(t)\not\in U$ for all $t\ge 0$. We have shown that $x\in \Wd$ for all $x\in \B_r(x_0)$. Hence $\Wd$ is open. \hfill \BlackBox

\section{Proof of Proposition \ref{prop:clf}}\label{sec:clf}

The existence of a Lyapunov function can be proved based on the $\mathcal{KL}$-stability (i.e. given in Definition \ref{defn: KL-stab}), following the techniques developed in \cite{teel2000smooth} on converse Lyapunov functions for $\mathcal{KL}$-stability. The $\mathcal{KL}$-stability considered here is in fact a special case of that in \cite{teel2000smooth}, because we do not need to consider stability with respect to two different measures as in  \cite{teel2000smooth}. We provide a definition of $\mathcal{KL}$-stability below, adapted for a proper indicator of a compact set. 

\begin{defn}\label{defn: KL-stab}
Let $A\subseteq \mathbb{R}^n$ be a compact set contained in an open set $D\subseteq \mathbb{R}^n$. Let $\omega$ be any proper indicator for $A$ on $D$. The system $\mathcal{S}_{\delta}$ is said to be $\mathcal{KL}$-stable on $D$ w.r.t. $\omega$ if 
any solution $\phi\in\Sd(x)$ with $x\in D$ is defined and remain in $D$ for all $t\ge 0$ and there exists a $\mathcal{KL}$-function $\beta$ such that
\begin{equation}\label{E: eq_ym1}
\omega(\phi(t;x))\leq \beta (\omega(x),t),\;\;\;\;\forall t\geq 0, 
\end{equation}
for all $x\in D$ and $\phi\in\Pd(x)$. 
\end{defn}

The key step in proving Proposition \ref{prop:clf} is the following lemma. 

\begin{lem}\label{lma:KL-stability}
Assume that the assumptions of Proposition \ref{prop:clf} hold. Then the system $\mathcal{S}_{\delta}$ is $\mathcal{KL}$-stable on $D$ w.r.t. $\omega$.
\end{lem}

\textbf{Proof of Lemma \ref{lma:KL-stability}}
Let $\mathcal{C}_r:=\{x\in D: \omega(x)\leq r\}$. Then by the assumptions, since $\omega$  is a proper indicator $\omega$ for $A$ on $D$, $\mathcal{C}_r$ is compact subset of $D$ for each $r\geq 0$. Fix $\rho>0$ such that $A+\rho\mathcal{B}\subset D$. We can find a $\mathcal{K}_{\infty}$-class function satisfying $\alpha(s)\geq\sup_{x\in D,\|x\|_A\leq \min(\rho,s)}\omega(x)$. Therefore, for all $\|x\|_A\leq \rho$, we have $\omega(x)\leq\alpha(\|x\|_A)$. 

\begin{claim}\label{clm:clm_ym2}
There exists a $\mathcal{K}_{\infty}$ function $\gamma$ such that, for each $x\in D$, $\omega(\phi(t;x))\leq \gamma(\omega(x))$ for all $t\geq 0$ and $\phi\in\Pd(x)$.
\end{claim}

\textbf{Proof of Claim \ref{clm:clm_ym2}} Indeed, for each $x\in D$, we can find an $r\geq 0$ such that $x\in\mathcal{C}_r$. %
By Proposition \ref{prop:uniform}, for any $\rho>0$ chosen above, we can find a $T$ such that $\|\phi(t;x)\|_A\leq \rho$ for all $x\in \mathcal{C}_r$ and $\phi\in\Pd(x)$. By forward invariance of $D$, it follows that $\mathcal{R}_{\delta}(\mathcal{C}_r)\subset\mathcal{R}_{\delta}^{0\leq t\leq T}(\mathcal{C}_r)\cup (A+\rho\mathcal{B})\subset D$. Since $\mathcal{C}_r$ is compact, by Lemma \ref{lem:compact}, for any finite $T$, $\mathcal{R}_{\delta}^{0\leq t\leq T}(\mathcal{C}_r)$ is also compact. The boundedness of $\mathcal{R}_{\delta}(\mathcal{C}_r)$ implies that $\overline{\mathcal{R}_{\delta}(\mathcal{C}_r)}$ is a compact subset of $D$. Let  $M(r)=\max_{x\in\overline{\mathcal{R}_{\delta}(\mathcal{C}_r)}}\omega(x)$. Then $\omega(\phi(t;x))\leq M(\omega(x))$ for all $x\in D$, $\phi\in\Pd(x)$, and $t\ge 0$. Clearly, $M(r)$ is nondecreasing (due to the inclusion relation of reachable sets from $\mathcal{C}_r$ with different $r$) and $\lim_{r\rightarrow 0}M(r)=0$ (due to the uniform stability property). The $\gamma\in\mathcal{K}_{\infty}$ in the claim can be chosen such that $M(r)\leq \gamma(r)$ for all $r\geq 0$. \qed

\begin{claim}\label{clm:clm_ym1}
For each $r>0$, there exists a strictly decreasing function $\psi_r :\mathbb{R}_{>0}\rightarrow\mathbb{R}_{>0}$ with $\lim_{t\ra\infty} \psi_r^{-1}(t)=0$ such that $\omega(\phi(t;x))\leq \psi_r^{-1}(t)$ for all $t>0$ whenever $\omega(x)\leq r$ and $\phi\in\Pd(x)$. 
\end{claim}

\textbf{Proof of Claim \ref{clm:clm_ym1}} 
For each $0<\varepsilon\leq \gamma(r)$, %
by Proposition \ref{prop:uniform}, we can find a $T_r(\eps)=T(\min(\alpha^{-1}(\varepsilon),\rho))>0$ such that for all $x\in\mathcal{C}_r$, $\phi\in\Pd(x)$, and $t\geq T_r(\varepsilon)$, we have
\begin{equation}\label{E: claim27}
  \|\phi(t;x)\|_A<\min(\alpha^{-1}(\varepsilon),\rho)\leq \rho. 
\end{equation}
Equation \eqref{E: claim27} also implies 
\begin{equation}\label{eq:boundphi}
\omega(\phi(t;x))\leq \alpha(\alpha^{-1}(\varepsilon))=\varepsilon, 
\end{equation}
for all $x\in\mathcal{C}_r$, $\phi\in\Pd(x)$, and $t\geq T_r(\varepsilon)$. For $\varepsilon>\gamma(r)$,  we set $T_r(\eps)=0$ and (\ref{eq:boundphi}) still holds because $\omega(\phi(t;x))\le \gamma(r) < \eps$ for all $t\ge 0$ by Claim \ref{clm:clm_ym2}. Note that for each fixed $r$, the function $T_r(\eps)$ can be chosen to be nonincreasing in $\varepsilon$ and by definition $\lim_{\varepsilon\rightarrow +\infty}T_r(\varepsilon)=0$; for each fixed $\varepsilon>0$, $T_{r}(\varepsilon)$ can be chosen to be nondecreasing in $r$. Based on $T_r(\eps)$, we can find $\psi_r:\mathbb{R}_{>0}\rightarrow\mathbb{R}_{>0}$ such that $\psi_r(\varepsilon)\geq T_r(\varepsilon)$ for all $\varepsilon>0$. The function $\psi_r$ can be constructed as strictly decreasing to zero (hence its inverse is defined on $\Real_{>0}$ and also strictly decreasing) and satisfying $\lim_{t\ra\infty} \psi_r^{-1}(t)=0$.  For each $t>0$, let $\varepsilon=\psi_r^{-1}(t)$. We have $t=\psi_{r}(\eps)\geq T_r(\varepsilon)$. Hence $x\in\mathcal{C}_r$ implies that $ \omega(\phi(t;x))\leq \eps=\psi_r^{-1}(t)$.  \qed

Now we force $\psi_r^{-1}(0)=\infty$ defined in Claim \ref{clm:clm_ym1} and let $\beta(s,t):=\min\{\gamma(s),\inf_{r\in(s,\infty)}\psi_r^{-1}(t)\}$ with $\gamma$ defined in Claim \ref{clm:clm_ym2}. Then $\beta\in\mathcal{KL}$\footnote{This construction of $\mathcal{KL}$ function does not impose continuity. Nonetheless, as pointed out in \cite[Remark 3]{teel2000smooth}, any (potentially noncontinuous) $\mathcal{KL}$ function can be upper bounded by a continuous $\mathcal{KL}$ function.} and \eqref{E: eq_ym1} holds. \hfill \BlackBox

Once Lemma \ref{lma:KL-stability} is proved, the proof of Proposition \ref{prop:clf} follows from a standard converse Lyapunov argument (see \cite[proof of Theorem 1]{teel2000smooth}). We provide an outline of the proof as follows. 

\begin{lem}[Sontag \cite{sontag1998comments}]\label{lma: sontag}
For each $\beta\in \mathcal{KL}$ and each $\lambda>0$, there exist functions $\alpha_1,\alpha_2\in\mathcal{K}_{\infty}$ such that $\alpha_1$ is locally Lipschitz and
\begin{equation}
\alpha_1(\beta(s,t))\leq \alpha_2(s)e^{-\lambda t},\;\;\;\;\forall (s,t)\in\mathbb{R}_{\ge 0}\times\mathbb{R}_{\ge 0}.
\end{equation}
\end{lem}

\textbf{Proof of Proposition 10}
 Based on the Lemma \ref{lma:KL-stability} and by Sontag's lemma (Lemma \ref{lma: sontag}) on $\mathcal{KL}$-estimates, we can find $\alpha_1$ and $\alpha_2$ such that
\begin{equation}\label{E:sontag}
\alpha_1(\omega(\phi(t;x)))\leq \alpha_1(\beta(\omega(x),t))\leq \alpha_2(\omega(x))e^{-2t}
\end{equation}
for any $x\in D$, $\phi\in\Pd(x)$, and $t\ge 0$. Now define
\begin{equation}\label{E: V_construction}
V(x):=\sup\limits_{t\geq 0, x\in D, \phi\in\Pd(x)}\alpha_1(\omega(\phi(t;x)))e^t.
\end{equation}
Then $V(x)\geq \sup_{\phi\in\Pd(x)}\alpha_1(\omega(\phi(t;x)))=\alpha_1(\omega(x))$ for all $x\in D$, and it is straightforward from \eqref{E:sontag} that $V(x)\leq \sup_{t\geq 0}\alpha_2(\omega(x))e^{-t}\leq \alpha_2(\omega(x))$. Therefore condition $(4)$ in Theorem $8$ is satisfied.

To show the satisfaction of condition $(5)$ in Theorem \ref{thm:conv1}, we can %
show that
\begin{equation}
V(\phi(t;x))\leq V(x)e^{-t},\quad \forall \phi\in\Pd(x),\forall t\ge 0,
\end{equation}
with a similar reasoning as the Claim 1 in \cite{teel2000smooth}. The local Lipschitz continuity of $V$ follows from the Claim 3 in \cite{teel2000smooth}. Then we have
\begin{equation}
\begin{split}
\nabla V(x)\cdot (f(x)+d) & \leq \liminf\limits_{t\rightarrow 0^+}\frac{V(\phi(t;x,d))-V(x)}{t}\\
& \leq \liminf\limits_{t\rightarrow 0^+}V(x)\frac{e^{-t}-1}{t}=-V(x).
\end{split}
\end{equation}
The smoothness of $V$ can also be extended from the locally Lipschitz region $D\setminus A$ to the whole set $D$ (by following the proof of Theorem 1 (step 3) in \cite{teel2000smooth}).  \hfill \BlackBox

\end{document}